\documentclass[12pt]{article}
\usepackage{mathptmx}
\usepackage{amssymb}
\usepackage{amsthm}
\usepackage{amsmath}
\usepackage{lscape}
\newtheorem{theorem}{Theorem}

\usepackage{float}
\usepackage{wrapfig}
\usepackage{natbib}
\usepackage{ascmac}
\usepackage{amsmath,amssymb}
\usepackage{amsthm}

\usepackage{latexsym}
\usepackage{subfigure}
\usepackage{bm}

\usepackage{graphicx}

\def\newblock{\hskip .11em plus .33em minus .07em}

\usepackage{multirow}
\usepackage{bigstrut}
\usepackage{endnotes}
\usepackage{url} 

\begin{document}
\setlength{\abovedisplayskip}{4pt}
\setlength{\belowdisplayskip}{4pt}
\setlength{\baselineskip}{7.1mm}
\title{Identification of causal direct-indirect effects \\ without untestable assumptions
\\
}

\date{}
\maketitle

\vspace{-1.5cm}

\begin{center}

{\large {\bf Takahiro Hoshino } \\
}
\end{center}

\vspace{1cm}

\thispagestyle{empty}

TAKAHIRO HOSHINO \\ Department of Economics, Keio University \\ 
RIKEN Center for Advanced Intelligence Project \\
E-mail: bayesian@jasmine.ocn.ne.jp



\begin{center} ABSTRACT \end{center}

In causal mediation analysis, identification of existing causal direct or indirect effects requires untestable assumptions in which potential outcomes and potential mediators are independent.
This paper defines a new causal direct and indirect effect that does not require the untestable assumptions.
We show that the proposed measure is identifiable from the observed data, even if potential outcomes and potential mediators are dependent, while the existing natural direct or indirect effects may find a pseudo-indirect effect when the untestable assumptions are violated.

\vspace{10mm}


\newpage

\pagestyle{plain}
\setcounter{page}{1}

\section{Introduction}

In recent years, there has been considerable methodological development and applied studies based on potential outcome approaches \citep{Rubin:jep1974} to causal mediation analysis to understand causal mechanisms (for example, \citealt{Pearl:2001,VanderWeele:2009,Imai:2010a,Tchetgen:2012,Ding:2016,Miles:2020}).
Let $T$ denote the exposure or treatment of interest, $Y$ the outcome, $M$ the mediator, and the baseline covariates $v$, which are not affected by the exposure and mediator.
Following the potential outcome approach, let $Y_{j(m)}$ be the potential outcome when $T=j$ and $M=m$.

Most recent studies in causal mediation analysis consider the natural direct/indirect effects, which are defined using the expectation of the ``never-observed'' outcome (not ``potential outcome'') $Y_{j}(M_k) \ \ (j \neq k)$, which is the potential outcome under treatment $j$ for the potential mediator for treatment $k$, $M_k$.
To identify these effects various assumptions are proposed. The following assumptions \citep{Pearl:2001} are often made:
\begin{description}
\item[Assumption 1] $Y_{j} (k) \perp\!\!\!\perp T  | v, \ \forall j,k,$
\item[Assumption 2] $M_{j} \perp\!\!\!\perp T  | v, \ \forall j,$
\item[Assumption 3] $Y_{j} (k) \perp\!\!\!\perp M| T, \ v, \ \forall j,k,$
\item[Assumption 4] $Y_{j} (k) \perp\!\!\!\perp M_{j^*} |v, \ \forall j,j^*,k,$
\end{description}
Another example for sufficient conditions for identifying the two natural effects includes the following sequential ignorability conditions (SI1 and SI2) by \citet{Imai:2010b}:
\begin{description}
\item[Assumption SI1] $\{Y_{j} (k), M_{j^*} \} \perp\!\!\!\perp T  | v, \ \forall j,j^*,k,$
\item[Assumption SI2] $Y_{j} (k) \perp\!\!\!\perp M_{j^*} | T=j^*, v, \ \forall j,j^*,k,$
\end{description}
For the relationship between these sets of conditions, see \citet{Pearl:2014}.

As has already been pointed out by various studies, Assumptions 1 and 2 or assumption SI1 are satisfied if $T$ is randomized. Assumption 4 or SI2 is not testable in that this states independence of potential outcomes and potential mediators, some of which we never observe simultaneously.
These assumptions do not hold even if both $T$ and $M$ are randomized or ignorable given $v$ \citep{Pearl:2014}, while Assumption 3 holds.

This paper defines new causal mediation effects that are identifiable from observed data without the untastable assumptions when both $T$ and $M$ are randomized or ignorable given $v$. The proposed ones are useful even if the ranzomization for the mediator is not possible in that the assumptions required for identification are weaker than those for the traditional estimands, natural direct/indirect effects.

The proposed direct and indirect effects will have the following properties:
\begin{description}
\item[(a)] indirect effect will be zero if $M_0=M_1$ for all units.
\item[(b)] these effects are identifiable without untestable Assumption 4 or SI2.
\item[(c)] the defined effects are expressed as the potential outcomes $Y_j(m)$ and the potential mediators $M_j$, not $Y_j(M_k)$. Thus, the causal interpretation is straightforward.
\end{description}

\section{Notation and existing estimand}

Without loss of generality, we consider binary treatment in this paper (for multi-valued treatment, we can generalise the result using similar arguments to those by \citealt{Imbens:2000}).
Using two potential mediators $M_1$ (for $T=1$ treatment) and $M_0$ (for $T=0$ treatment), $M$ is expressed as
\begin{equation}
\label{eqm}
M= T M_1 + (1-T) M_0.
\end{equation}
We assume that $M$ is a categorical variable (i.e., $M=0, \cdots , M^*$). 
Let $M(m)$ be the binary indicator such that $M(m)=1$ if $M=m$. Similarly, let $M_{j}(m)$ be the binary indicator under $T=j$ treatment, $M_{j}(m)=1$ if the potential mediator under treatment $j$, $M_{j}$, is $m$.

The observed outcome $Y$ is expressed by the potential outcomes, potential mediators, and treatment indicator as follows:
\begin{equation}
\label{eq1}
Y = \sum_{m=0}^{M^*} \Bigl[ T M(m) Y_{1}(m) + (1-T) M(m) Y_{0}(m) \Bigr] = \sum_{m=0}^{M^*} \Bigl[ T M_{1}(m) Y_{1}(m) + (1-T) M_{0}(m) Y_{0}(m) \Bigr]
\end{equation}
The observed outcome can be expressed by two potential outcomes $Y_1=Y_1(M_1)$ (for $T=1$ treatment) and $Y_0(M_0)$ (for $T=0$ treatment) as follows:
\begin{equation}
Y= T Y_1 + (1-T) Y_0 = T Y_1(M_1) + (1-T) Y_0(M_0)
\end{equation}
under the composition assumption \citep{Pearl:2009}.

The average treatment effect ($ATE$) is defined as the expectation of the difference between two potential outcomes:
\begin{equation}
\label{eq2}
ATE \equiv E[Y_1- Y_0], 
\end{equation}

A straightforward way of defining the direct effect is to set the mediator to a pre-specified level $M=m$. \citet{Pearl:2001} defined the controlled direct effect with mediator fixed at $M=m$, $CDE(m)$, in which the mediator is set to $m$ uniformly over the entire population:
\begin{equation}
ICDE(m) \equiv Y_{1}(m) - Y_{0}(m), \ \ \ CDE(m) \equiv E[ICDE(m)].
\end{equation}
where $ICDE(m)$ is an unit-level version of the $CDE(m)$ that we define here to use later in this paper. As pointed out in existing studies, the quantity defined as $ATE-CDE$ is not a proper measure of indirect effect in that this quantity may not be zero even when $M_0=M_1$ for all units \citep{VanderWeele:2009}

In the literature on causal mediation analysis \citep{Robins and Greenland:1992,Pearl:2001,Pearl:2009}, ATE is expressed as the sum of the natural direct effect ($NDE$) and the natural indirect effect ($NIE$), instead of using $CDE$. Following \citet{Imai:2010b},
\begin{align}
\label{eq3}
NDE(t) & \equiv  E[Y_{1}(M_t) - Y_{0}(M_t)], \ \ \ NIE(t) \equiv  E[Y_{t}(M_1) - Y_{t}(M_0)], \ (t=0,1)  \notag \\
NDE & \equiv \frac{1}{2} [NDE(1) + NDE(0) ] \ \ \ NIE \equiv \frac{1}{2} [NIE(1) + NIE(0) ]  \notag \\
ATE & = E[Y_{1} - Y_{0}] = E[Y_{1}(M_1) - Y_{0}(M_0)] = NDE + NIE.
\end{align}

Note that the NDE and NIE are not identified without further assumptions because quantity $Y_{1}(M_0)$ is not observable.
For identification, the existing studies assume independence between potential outcomes and mediator (given some observable covariates), or related conditions such as sequential ignorability. In Section 5, we show that NIE may be biased in that under no mediation, NIE is not zero when the assumption is violated while the proposed one is not.

\section{Definition of weighted direct effect and estimable indirect effect}

Identification of NDE and IDE requires assumption 4 or SI2 because the causal mediation effect is defined by using the ``never-observable''outcome (not ``potentially observable'' outcome) $Y_{j}(M_k) \ \ (j \neq k)$. 
However, $Y_{j}(k)$ and $M_{j}$ are observed for some portion of the units.

We redefine the potential outcomes $Y_{j}(M_k)$ by using the functions of potential outcomes and mediators as follows:
\begin{align}
\label{eq4}
Y_{j}(M_k)  \equiv \sum_{m=0}^{M^*} M_k(m) Y_{j}(m).
\end{align}
Under this definition, $Y_{j}(M_k) = Y_{j}(m)$ when $M_k(m)=1$.

By using this expression, ATE is determined as follows:
\begin{align}
ATE = E[Y_1(M_{1})- Y_0(M_{0})]=E[ \sum_{m=0}^{M^*} \{ M_1(m) Y_{1}(m) -  M_0(m) Y_{0}(m) \} ]
\end{align}
Then, we defined the {\bf weighted controlled direct effect } ($WCDE$)
\begin{align}
WCDE & \equiv E[ \sum_{m=0}^{M^*} \{ M(m) Y_{1}(m) -  M(m) Y_{0}(m) \} ] \notag \\
& = E[ \sum_{m=0}^{M^*} M(m) ( Y_{1}(m) -  Y_{0}(m) ) ] = E[ \sum_{m=0}^{M^*} M(m) ICDE(m) ].
\end{align}

Note that in $CDE(m)$, the mediator is set to be the specific value, $M=m$, while $WCDE$ is the weighted average of $ICDE(m)$ over the observed distribution of $M$.

The {\bf implied indirect effect} ($IIE$) is expressed as:
\begin{align}
IIE & \equiv ATE - WCDE  \notag \\
& = E[ \sum_{m=0}^{M^*} [(M_{1}(m)-M(m))Y_{1}(m) +  (M(m)-M_{0}(m))Y_{0}(m)]  ]
\end{align}

While the quantity defined as $ATE-CDE$ may not be zero even when $M_0=M_1$ for all units, $IIE$ is always zero if $M_{1}=M_{0}$ for all units.
\begin{theorem}
\label{theorem1}
$IIE$ is equivalent to zero if $M_{1}=M_{0}$ for all units.
\end{theorem}
\begin{proof}[of Theorem~\ref{theorem1}]
From Equation \ref{eqm}, if $M_{1}=M_{0}$ then IIE is zero because $M_{1}=M_{0}=M$.
\end{proof}

\subsection*{A case of a binary moderator}
We consider the case of a binary moderator. By Equation \ref{eq1}, the observed outcome $Y$ is expressed as
\begin{align}
\label{eq01}
Y &= T M Y_{1}(1) + T (1-M) Y_{1}(0) + (1-T) M Y_{0}(1) + (1-T)(1-M) Y_{0}(0) \notag \\
&= T M_1 Y_{1}(1) + T (1-M_1) Y_{1}(0) + (1-T) M_0 Y_{0}(1) + (1-T)(1-M_0) Y_{0}(0)
\end{align}
where $M_{1}(1)=M_1, M_{0}(1)=M_0, M_{1}(0)=1-M_1$, and $M_{0}(0)=1- M_0$.

Using Equation \ref{eq4},
\begin{align}
Y_1=Y_{1}(M_1) & =  M_1 Y_1(1) + (1-M_1) Y_1(0), \ \ Y_{1}(M_0) =  M_0 Y_1(1) + (1-M_0) Y_1(0) \notag \\
Y_{0}(M_1) & =  M_1 Y_0(1) + (1-M_1) Y_0(0), \ \ Y_0=Y_{0}(M_0) =  M_0 Y_0(1) + (1-M_0) Y_0(0),
\end{align}

For example, the potential outcome if the unit recieves $T=1$ and $M_0=1$ will be $Y_1(1)$.

Then, $ATE$, $WCDE$, and $IIE$ are expressed as follows:
\begin{align}
ATE & = E[ M_1 Y_{1}(1) - M_0 Y_{0}(1) + (1-M_1) Y_{1}(0)   - (1 - M_0) Y_{0}(0) ) ]  \notag \\
WCDE & = E[ M \times ICDE(1) + (1-M) \times ICDE(0) ]  \notag \\
& = E[ M ( Y_{1}(1) - Y_{0}(1) ) + (1-M) ( Y_{1}(0) - Y_{0}(0) ) ]  \notag \\
& = E[ (T M_1 + (1-T) M_0 ) ( Y_{1}(1) - Y_{0}(1) ) + (1- T M_1 - (1-T) M_0 ) ( Y_{1}(0) - Y_{0}(0) ) ]  \notag \\
& \neq NDE= \frac{1}{2} [NDE(1) + NDE(0)]  \notag \\
& = E[ \frac{M_1 + M_0}{2} ICDE(1) + \bigl( 1- \frac{M_1 + M_0}{2} \Bigr) ICDE(0) ]
\notag \\
IID & = E[ (M_1 -M) (Y_{1}(1) - Y_{1}(0)) + (M-M_0) (Y_{0}(1) - Y_{0}(0)) ] \notag \\
& \neq  NIE=\frac{1}{2} [NIE(1) + NIE(0)] =\frac{1}{2} E[ (M_1 - M_0) \{ (Y_{1}(0) - Y_{1}(0)) + (Y_{0}(0) - Y_{0}(0)) \}]
\end{align}
where $M(1)=M$ and $M(0)=1-M$.

It is easily shown that with randomization of $T$, the natural direct effect $NDE=\frac{1}{2} [NDE(1) + NDE(0)]$ evaluates the direct effect if the distribution of treatment is to be $p(T=1)=0.5$, which is different from the ``natural'' observed distribution.
From these equations. it is expected that the difference between WCDE and NDE will be larger when $P(T=1)$ deviates from $0.5$.
Note that in the above equations the expectation is taken over the population distribution of $Y_{j}(k) \ (j,k=0,1), M_{1},M_{0}$ and $T$. 
Moreover, as will be mentioned in the next section, $WCDE$ is estimable without Assumption 4 or SI2, while NDE is not.

\section{Identification and Estimation}

Instead of Assumptions 1-4 (or SI1 and SI2), we introduce the following mean independence versions of Assumptions 1-4, SI1 and SI2:
\begin{description}
\item[Assumption $1^{'}$] $E[Y_{j}(k)|T,v]=E[Y_{j}(k)|v] \ \forall j, k,$
\item[Assumption $2^{'}$] $E[M_{j}|T,v]=E[M_{j}|v] \ \forall j,$
\item[Assumption $3^{'}$] $E[Y_{j}(k)|M,T=j,v]=E[Y_{j}(k)|T=j,v] \ \forall j, k,$
\item[Assumption $4^{'}$] $E[Y_{j}(k)|M_{j^*},v]=E[Y_{j}(k)|v] \ \forall j, j^*, k,$
\item[Assumption $SI1^{'}$] $E[M_{j^*}Y_{j}(k)|T,v]=E[M_{j^*}Y_{j}(k)|v] \ \forall j, j^*, k,$
\item[Assumption $SI2^{'}$] $E[Y_{j}(k)|M_{j^*},T=j,v]=E[Y_{j}(k)|T=j,v] \ \forall j, j^*, k,$
\end{description}
Assumption $SI1^{'}$ implies the mean independence version of the ignorability assumption \citep{Rosenbaum:1983}, $E[Y_{j}|T,v]=E[Y_{j}|v] \ \forall j$, which is sufficient for identifying ATE.

For identification of $WCDE$, we consider the following two cases: Case 1, when both $M$ and $T$ are randomized or ignorable given $v$, and Case 2, when $M$ is not directly maipulable.

\subsection*{Case1: When both $M$ and $T$ can be randomized or ignorable given covariates}
In this case we can identify ATE, WCDE and IIE in the following ways:
\begin{description}
\item[Step 1] Divide the sample into two equivalent subgroups (usually by using randomization).
\item[Step 2] Randomize $T$ with $v$ given in the first group to obtain a consistent estimator of ATE, $\hat{E}[Y_1 -Y_0] $, and that of $p(M)=p(M_1|T=1)p(T=1) + p(M_0|T=0)p(T=0)$.
\item[Step 3] Randomize both $M$ and $T$ with $v$ given in the second group, in which the distributions of $T$ and $M$ are set to be equal to those in the first group to identify WCDE (and IID) by using Theorem~\ref{theorem2} below.
\end{description}
Note that in the first group by randomizing $T$, Assumptions $2^{'}$ and $SI1^{'}$ automatically hold, which is sufficient to identify $ATE$ and $p(M)$. In the second group, by randomizing both $T$ and $M$, Assumptions $1^{'}$ and  $3^{'}$ hold.
From the following theorem, using the data from the second group $WCDE$ and $IID=ATE-WCDE$ are identifiable without any additional assumptions such as mean independence between potential outcomes and potential mediators (i.e., Assumption $4^{`}$ or $SI2^{`}$).
\begin{theorem}
\label{theorem2}
WCDE is identifiable by observed data under Assumptions $1^{'}$ and $3^{'}$.
\end{theorem}
\begin{proof}[of Theorem~\ref{theorem2}]
Under these assumptions, the $WCDE$ is expressed as
\begin{align}
WCDE & =  E_v \Bigl[ \sum^{M^*}_{m=0} E[ M(m) ICDE(m)  |v]  \Bigr] = E_v \Bigl[ \sum^{M^*}_{m=0} E_T[ E[M(m)|T]E[ ICDE(m) |T] |v]  \Bigr] \notag \\
& = E_v \sum^{M^*}_{m=0} \Bigl[ p(T=1|v) p(M(m)=1|T=1,v) E[ Y_{1}(m) - Y_{0}(m)|T=1,v]  \notag \\
&+ p(T=0|v) p(M(m)=1|T=0,v) E[ Y_{1}(m) - Y_{0}(m)|T=0,v]  \Bigr] \notag \\
& = E_v \sum^{M^*}_{m=0} \Bigl[ p(M(m)=1|v) ( E[ Y_{1}(m)|v] - E[Y_{0}(m)|v] ) \Bigr] \notag \\
& = E_v \sum^{M^*}_{m=0} \Bigl[ p(M(m)=1|v) \times ( E[ Y_{1}(m)|T=1,M=m,v] - E[Y_{0}(m)|T=0,M=m,v] )  \Bigr]
\end{align}
Considering that $p(M)$, $E[ Y_{j}(k)|T=j,M=k,v]  \ (\forall j,k)$ are observable, the $WCDE$ is identifiable.
\end{proof}

It should be noted again that the identification of NDE and NIE requires Assumption $4^{'}$ or $SI2^{'}$ even after randomizing both $T$ and $M$ \citep{Pearl:2014}.

\subsection*{Case2: When $M$ is not directly manipulable}

If randomization of $M$ is not feasible, it is inevitable to accept some untestable assumptions to identify causal direct/indirect effects. As stated in Theorem 2, it is sufficient to assume Assumptions $1^{'}, 2^{'}, 3^{'}$ and $SI1^{'}$ hold given abundant covariates to identify WCDE and IID.
In this case, $M$ is not directly manipulatable, then Assumption $3^{'}$ (and the other assumptions when $T$ is also not manipulatable) is untestable, but as mentioned earlier, these assumptions are weaker than Assumption $4^{'}$ or $SI2^{'}$
.
\subsection*{Estimation}

For simplicity, we consider the case without covariates.
For Case 1, the estimator of WCDE is expressed by the observed quantities:
\begin{align}
\widehat{WCDE} = \sum^{M^*}_{m=0} \hat{p}(M(m)=1)(\bar{y}_{1|M=m} - \bar{y}_{0|M=m} )  
\end{align}
where $\hat{p}(M(m)=1)$ is the sample proportion with $M=m$ in the first group and $\bar{y}_{j|M=m}$ is the average of $y$ for units with $T=j$ and $M=m$ in the second group.
By simple application of the Delta method the asymptotic variance of $\widehat{WCDE}$ is expressed as
\begin{align}
\frac{1}{N} \Bigl\{  d^t (diag (p) - p p^t ) d + \sum^{M^*}_{m=0} p_m^2 [V(\bar{y}_{1|M=m}) + V(\bar{y}_{0|M=m}) ]  \Bigr\} 
\end{align}
where $p_{m}=P(M(m)=1)$, $p=(p_{0}, \cdots , p_{M^*})^t$, $d_m=E(y_{1}(m)) - E(y_{0}(m))$ and $d=(d_{0}, \cdots , d_{M^*})^t$.

The unbiased ATE is $\bar{y}_{1}- \bar{y}_{0}$, where $\bar{y}_{j}$ is the average of $y$ for units with $T=j$ in the whole sample, because even in the second group Assumptions $1^{'}$ and  $3^{'}$ hold, then the difference of averages of outcomes is an unbiased estimator of ATE in the second group.

For Case 2, under ignorability given covariate $v$, ATE and WCDE are expressed as:
\begin{align}
ATE &= E_v \Bigl[ E[Y|T=1,v] - E[Y|T=0,v] \} \Bigr] \notag \\
WCDE & = E_v \sum^{M^*}_{m=0} \Bigl[ P(M(m)=1|v) \{ E[Y|T=1,M=m,v] - E[Y|T=0,M=m,v] \}   \Bigr]
\end{align}
Then various methods such as inverse probability weighting estimator or doubly robust type estimator can be used to estimate $ATE$ and $WCDE$.

\subsection*{Hypothetical ratio adjustment for treatment}

As stated in the previous section, ``WCDE'' evaluates the direct effect with the observed proportion ($ \hat{p}(T=1)$) of the treatment group. If the researcher needs to consider the direct effect with a hypothetical proportion (say $p^*$) of the treatment group, use a weight of $\frac{p^*}{\hat{p}(T=1)}$ for treatment individual and use a weight of $\frac{1-p^*}{1-\hat{p}(T=1)}$  for control individual.

\subsection*{Generalization}

We can address the case where $m$ is continuous. Under continuous mediator $m$, WCDE is expressed as follows:
\begin{align}
WCDE = \int \int \{ E[Y_{1}(m)|m,v] - E[Y_{0}(m)|m,v] \} p(m|v) p(v) d m d v 
\end{align}
Under Assumptions $1^{'}$ and $3^{'}$, WCDE is expressed as the following quantities identifiable by observed data:
\begin{align}
WCDE = \int \int \{ E[Y_{1}(m)|T=1,m,v] - E[Y_{0}(m)|T=0,m,v]  \} p(m|v) p(v) d m d v
\end{align}
ATE is identifiable by Assumption $SI1^{'}$, then IID is also identified as $IID=ATE-WCDE$.

\section{Illustrative simulation}

For illustrative purposes, we present a simulation study that compares the defined effects with the previously proposed ones.
We numerically evaluated bias from the true values (population WCDE for the proposed method in Case 1 and population NDE for the existing method with assumptions SI1 and SI2).
We consider the data-generating model in which assumption SI2 in Section 1 can be violated.

For simplicity, we consider binary mediator $M$, and define two latent continuous potential mediators $M_0^L$ and $M_1^L$ so that $M_j=1$ if $M_j^L > 0$ otherwise $M_j=0 \ \ (j=0,1)$.
We generated 10,000 samples of size $n=4,000$ from the joint vector of potential mediators and potential , $W=(M_0^L, M_1^L, Y_{0}(0), Y_{0}(1), Y_{1}(0), Y_{1}(1))^t$ which follows a finite scale mixture of multivariate-normal distributions;\begin{align}
W \sim  0.6 \times N(\mu, \Sigma_1) + 0.4 \times N(\mu,\Sigma_2)
\end{align}
where $\mu=(-1,1,0,0.2,0.6,1)$, $\Sigma_1=\Sigma$, $\Sigma_2 = 2 \times \Sigma$ and the diagonal elements of $\Sigma$ is set to be $1$. Off-diagonal elements are set such that $Cov(Y_{j}(k)),Y_{l}(m))=0.5$ and $Cov(M_0,M_1)=0.6$.
For simplicity, the covariances between potential mediators and potential outcomes are set as follows: $Cov(M_0, Y_{0}(0))=Cov(M_0, Y_{1}(0)=Cov(M_1, Y_{0}(1))=Cov(M_1, Y_{1}(1))=\phi$, and $Cov(M_0, Y_{0}(1))=Cov(M_0, Y_{1}(1))=Cov(M_1, Y_{0}(0))=Cov(M_1 Y_{1}(0))=-\phi$. In this setup, $p(M_1=1) \approx 0.809$ and $p(M_0=1) \approx 0.191$.

The treatment indicator $T$ is generated from the Bernoulli distribution with $Pr(T=1)= p$, independently of $W$. 
We consider four cases with $p=0.01$, $0.1$, $0.3$ and $0.5$ and in each case $\phi$ varies from $-0.15$ to $0.15$ in $0.05$ increments.
Note that for all the models assumption 4 or SI2 is violated, except  $\phi=0$.

The true values of $WCDE$ and $NDE$ of the model is difficult to calculate analytically. We then evaluate these values using the simulated 10,000 datasets.

We compare the proposed estimatior for $WCDE$ with the existing estimator of $NDE$ implied by Equation 18 in \citet{Imai:2010b}, which is frequently used in applied studies.
The results are shown in Table 1 and those with $p=0.5$ are illustrated in Figure 1, in which the horizontal axis is $\phi$ and the vertical axis is the bias from the true values.

As shown in this figure and Table1, the bias of the previously proposed estimator can be large for large deviations from Assumption SI2 (i.e., large $\phi$) as mentioned in the sensitivity analysis in \citet{Imai:2010b}, while the proposed method can find true values on average.
The tendency of the size of bias does not change according to $p$, the proportion of $T=1$, but in setups with small $p$ the variance of the two estimators is large because the sample size with $T=1$ is very small.

Table 1 shows the numerically evaluated population $ATE$, $WCDE$ and $NDE$ in each setup. The difference between $WCDE$ and $NDE$ for small $p$ exists but not large, while as mentioned in Section 3 the difference is negligible if $p(T=1)=0.5$.

It is also shown that the bias of the existing estimator can be very large compared with the size of true value of the causal mediation effect.
In particular, as seen in the setup with $\phi=-0,15$ where $ATE$ is almost the same as $NDE$ (i.e., $NIE$ is almost zero), the existing method wrongly finds a ``(pseudo-)mediation'' effect under the setup when the true model does not contain a mediation effect, but the mediator and potential outcomes are correlated.

\begin{center}
Tab.1. \ \ Simulation results\label{tab1}
\includegraphics[width=14cm]{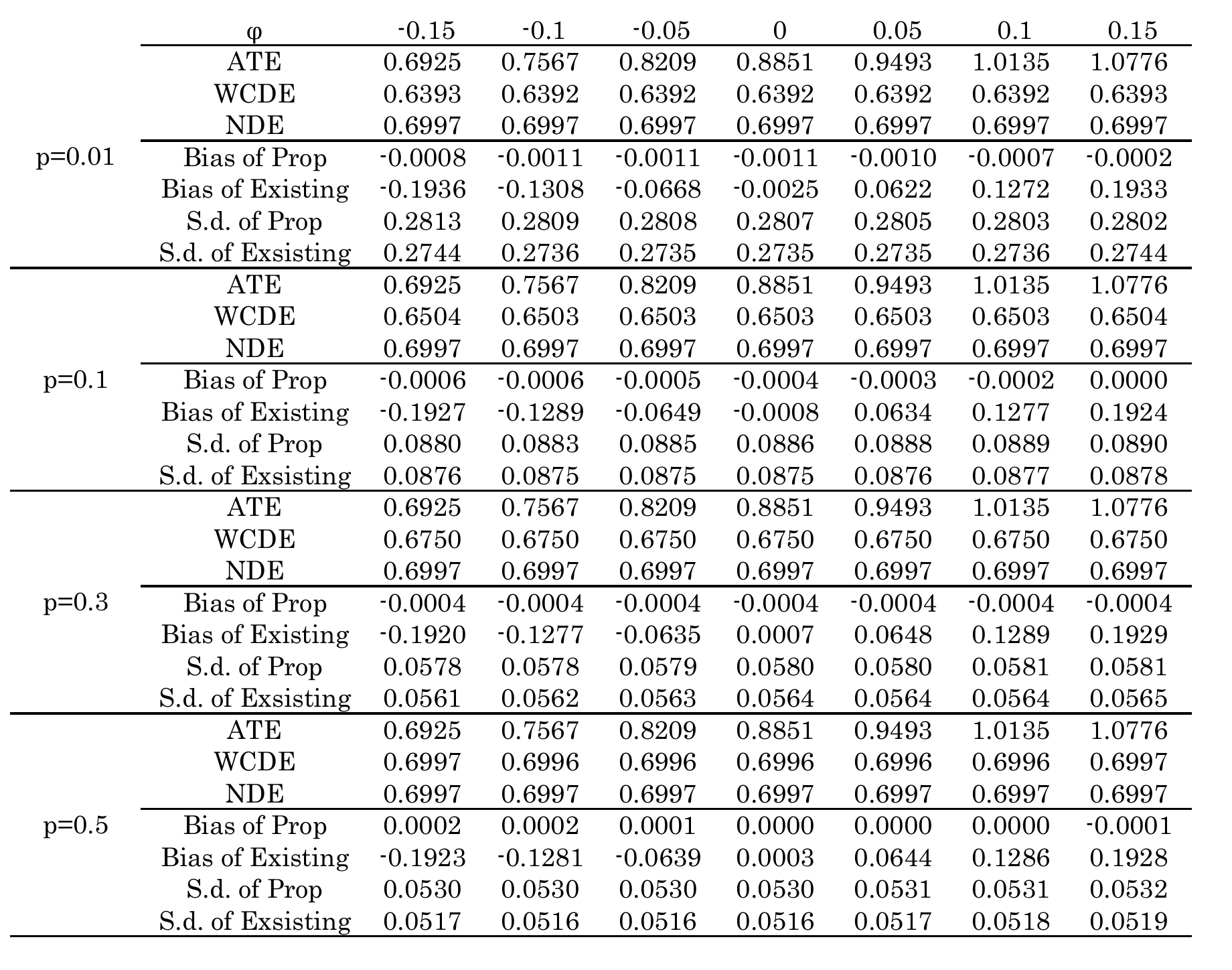}
\end{center}

\begin{figure}
\begin{center}
\caption{Result for $p=0.5$}\label{fig1}
\includegraphics[width=12cm]{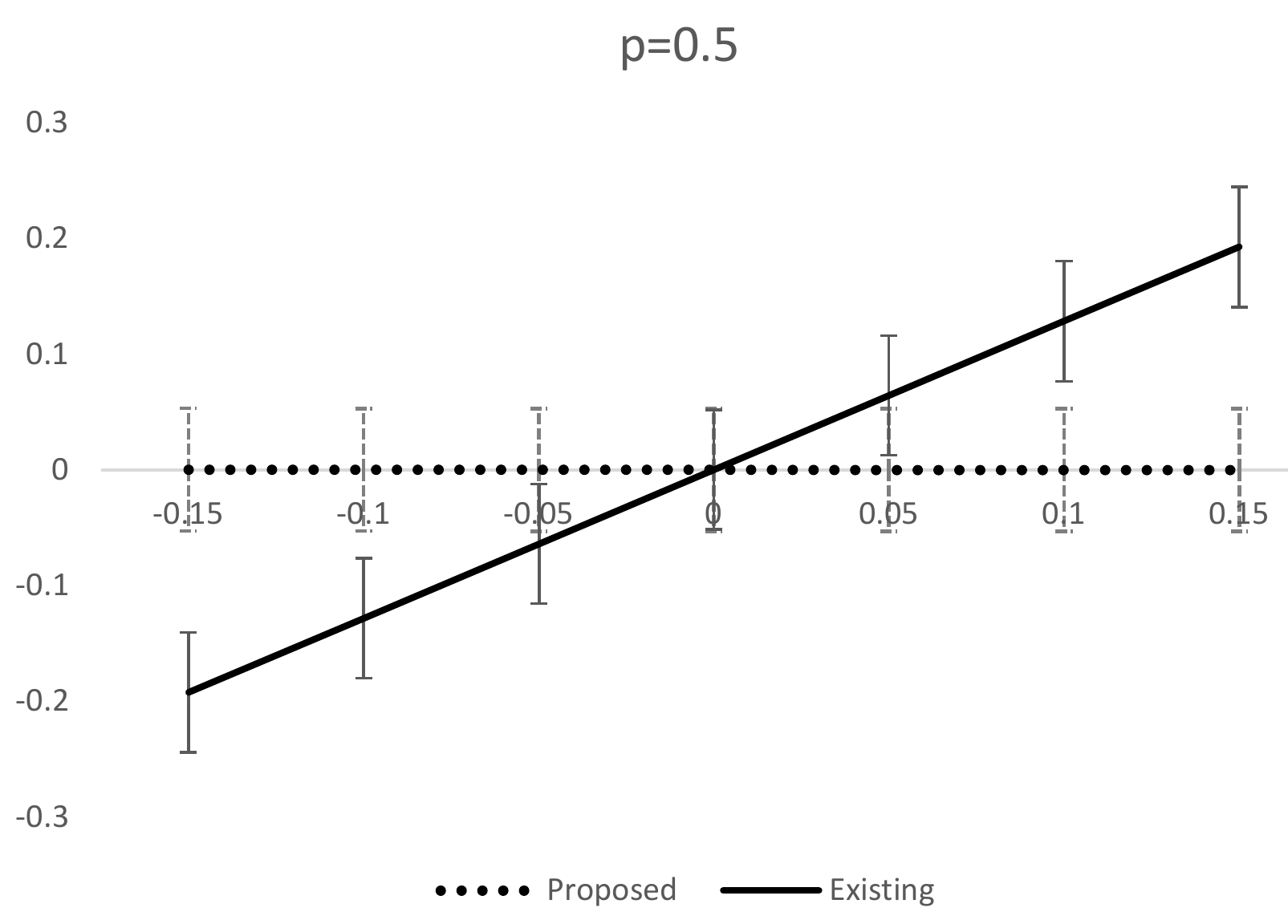}
\end{center}
The error bar indicates one standard deviation calculated from the 10,000 estimates.
\end{figure}

\section{Discussion}

In this paper, we proposed a new definition of causal direct and indirect effects in causal mediation analysis.

Identification of the previously proposed quantities, natural direct effect, and natural indirect effect is not possible even when both treatment and mediator are randomized. Therefore, it is unavoidable to employ untestable assumption of the independence of potential outcomes and potential mediators, some of which we never observe simultaneously.

The proposed quantities are identifiable without any assumption when both treatment and mediator are randomized. Even when randomization is not possible, the proposed quantities require weaker assumptions than those for the identification of traditional quantities.

When it is difficult to directly manipulate $M$, Assumption $3^{'}$ is required for identification of the proposed effects. Another approach for identification such as principal stratification approach \citep{Frangakis:2002,Forastiere:2018} is an promising strategy that we will investigate in a future study.

In this paper, we focused on the case with binary treatment, but the definition of WCDE is useful for the case with multi-valued treatment. In multi-valued treatment, the measure of indirect effect should be defined in terms of the sum of squares or variance of the expectations, instead of the traditional ``difference of the expectations'' formulation in the case of  binary treatment, which will be considered elsewhere.


\end{document}